\documentclass[11pt]{IEEEtran}

\def\BibTeX{{\rm B\kern-.05em{\sc i\kern-.025em b}\kern-.08em
    T\kern-.1667em\lower.7ex\hbox{E}\kern-.125emX}}

\usepackage{times}
\usepackage{subfigure}
\usepackage[dvips]{graphics,graphicx}
\DeclareGraphicsExtensions{.jpg,.pdf,.eps,.png}

\newtheorem {lemma}{Lemma}

\begin{document}

\title{Two channel paraunitary filter banks based on linear canonical transform}

\author{Sudarshan Shinde \\
        Computational Research Laboratories, \\
        Pune,INDIA. \\
Email:sudarshan\_shinde@iitbombay.org}

%\thanks{This work was supported by Naval Research Board of India in the project entitled ``Transient Detection using Wavelets''.}}

\maketitle

%\markboth{IEEE Transactions On Automatic Control, Vol. XX, No. Y, Month 1999}
%{Murray and Balemi: Using the style file IEEEtran.sty}

%\noindent EDICS No : 2-SDES Signal Detection And Estimation.

%\thispagestyle{plain}
%\pagestyle{plain}

\begin{abstract}
In this paper a two channel paraunitary filter bank is proposed , which is based on linear canonical transform, instead of discrete Fourier transform. Input-output relation for such a filter bank are derived in terms of polyphase matrices and modulation matrices. It is shown that like conventional filter banks, the LCT based paraunitary filter banks need only one filter to be designed and rest of the filters can be obtained from it. It is also shown that LCT based paraunitary filter banks can be designed by using conventional power-symmetric filter design in Fourier domain.   
\end{abstract}


\section{Introduction}
\subsection{Motivation}

Filter banks are now well known tool for time-frequency analysis of a signal \cite{bk_vaidyanathan_filterbank,bk_vetterli_filterbank}. Given a signal $x(n)$, a two channel filter bank (as shown in fig.\ref{fig_fb}), splits its Fourier transform (FT) spectrum into two by a low-pass and  a high-pass filter $\{h_{0}(n), h_{1}(n)\}$ respectively. Output of these filters are then down-sampled to produce analysis filter bank output $\{y_{0}(n),y_{1}(n)\}$. To reconstruct the signal $\{y_{0}(n),y_{1}(n)\}$ are passed from a synthesis filter bank. The synthesis filter bank upsamples $\{y_{0}(n),y_{1}(n)\}$, and after passing them from synthesis filters $\{g_{0}(n),g_{1}(n)\}$, adds them to form the filter bank output $\hat{x}(n)$. The filter bank introduces various distortions like alias distortion, magnitude distortion and phase distortion into the signal. It can be shown that by careful design of the  analysis and synthesis filters, all these distortions could be eliminated and $\hat{x}(n)$ could be made a delayed version of $x(n)$ by some integer $K$. Such filter banks are called perfect reconstruction filter banks (PRFB)\cite{bk_vaidyanathan_filterbank,bk_vetterli_filterbank}.

Discrete Fourier transform(DFT), and its generalization z-transform (ZT) find extensive use in the filter bank theory. Many key theorems of filter bank theory are expressed in Fourier domain or z-domain. Terms like low-pass and high-pass make sense in the Fourier domain. Design of the filters is carried out by giving their specifications in the Fourier domain. 

In another development, the Fourier transform(FT) has been generalized to fractional Fourier Transform (FrFT), by looking at FT as a transform that rotates the signal in time-frequency plane by $\pi/2$ \cite{bk_ozaktas,art_almeida_frft}. It is found that the FrFT is a special case of three parameter family of transforms called Linear Canonical Transform(LCT)\cite{art_stern_samplingLCT,bk_wolf_integralt}. Given a continuous time signal $x(t)$, its LCT is given by

\begin{equation}
X_{a,b,c,d}(u) = \sqrt{\frac{1}{j2\pi b}}\int_{\infty}^{\infty}{x(t)e^{j \left(\frac{a}{2b}t^{2} - \frac{1}{b}ut + \frac{d}{2b} u^{2} \right)}} dt
%X_{a,b,c,d}(u) = \sqrt{\frac{1}{j2\pi b}} \int_{\infty}^{\infty}{x(t)e^{j \left(\frac{a}{2b}t^{2} - \frac{1}{b}ut + \frac{d}{2b} u^{2} \right)}
    \label{eq_cont_lct} 
\end{equation}     
where $ad-bc = 1$.

In order to define FrFT and LCT on a discrete signal, many sampling theorems in FrFT domain and LCT doamin have been derived \cite{art_li_sampling_lct,art_candan_series,art_stern_samplingLCT,art_tao_samplingFRT,art_Zhao_convertLCT}. It is shown \cite{art_stern_samplingLCT}  that if a signal is band-limited in the LCT domain, then it can be reconstructed by its uniformly sampled discrete samples. Based on uniform sampling, a discrete time LCT (DTLCT) is introduced in \cite{art_Zhao_convertLCT}, and sampling rate conversion theorems for this DTLCT are also derived. 

Since sampling rate conversion is used in the filter banks, it is natural to ask if a sub-band decomposition scheme could be developed which is based on the DTLCT. In particular, we want to know how a filter-bank can be designed that splits the DTLCT spectrum of an input signal into two sub-band signals.

\subsection{Contributions of the paper}
In this paper the sampling rate conversion theorems derived in \cite{art_Zhao_convertLCT} are used to develop a sub-band decomposition scheme based on DTLCT. A convolution is defined in the DTLCT domain which is suitable for the sub-band decomposition scheme proposed. Input-output relation of an LCT filter bank in polyphase and modulation (or alias component) domains is derived. Condition for a two channel LCT filter bank to be paraunitary (PU) is also derived. It is shown further that a two channel PU LCT filter bank can be designed from a power-symmetric filter in ZT domain.

\subsection{Notations}
Given a signal $x(n)$, its LCT is denoted by $X(\omega)$, FT by $X(e^{j\omega})$ and ZT by $X(z)$. Sampling rate of $x(n)$ is denoted by $T_{x}$. Downsampling $x(n)$ by an integer $N$ is denoted by $x(n) \downarrow N$, and upsampling by $N$ is denoted by $x(n) \uparrow N$. Matrices are denoted by boldface letters. Complex conjugate of $X(\omega)$ is denoted by $X_{*}(\omega)$. Conjugate transpose of a matrix ${\bf H}(\omega)$ is denoted by ${\bf H}_{*}^{T}(\omega)$. For a scalar $H(z)$, $H_{*}(z^{-1})$ is denoted by $\tilde{H}(z)$, and for a matrix ${\bf H}(z)$, ${\bf H}_{*}^{T}(z^{-1})$ is denoted by $\tilde{{\bf H}}(z)$.
  
\section{Review of Filter Banks and DTLCT}
\label{sec_review}
\subsection{Review of DTLCT}
Let $x(n)$ be a discrete signal, obtained by uniform sampling of a continuous signal $x_{c}(t)$ with a sampling period $T_{x}$, so that $x(n) = x_{c}(nT_{x})$. The DTLCT of $x(n)$, with parameters $(a,b,c,d)$ is defined as \cite{art_Zhao_convertLCT},  

\begin{equation}
X(\omega) = \sqrt{\frac{1}{j2\pi b}}
            \sum_{n \in I}{x(n)\exp\left[ \frac{j}{2} 
                                         \left(
                                           \frac{a}{b}n^{2}T_{x}^{2}
                                           -2n\mbox{sgn}(b)\omega
                                           +\frac{db}{T_{x}^{2}} \omega^{2}  
                                         \right) 
                                  \right]}, b \not= 0,
    \label{eq_dtlct} 
\end{equation} 
where $\mbox{sgn}(b)$ is sign of $b$. In this paper we will assume $b > 0$ to avoid $\mbox{sgn}(b)$ term. This does not result in loss of generality, since all the results in this paper can be obtained for $b < 0$ case in the same way they are obtained for $b > 0$ case.

Upsampling a signal by $L$ is defined as inserting $L$ zeros between two samples. This changes the sampling time of the upsampled signal. Thus if $y(n) = x(n) \uparrow L$, then $T_{y} = T_{x}/L$. It can be shown that \cite{art_Zhao_convertLCT}
\begin{equation}
Y(\omega) = X(L\omega)
    \label{eq_upsample} 
\end{equation} 

Similarly downsampling a signal by $M$ is defined as dropping $M-1$ samples out of a block of $M$ samples and retaining only the first sample. This also changes the sampling time of the downsampled signal. Thus if $y(n) = x(n) \downarrow M$, then $T_{y} = MT_{x}$. It can again be shown that \cite{art_Zhao_convertLCT}
\begin{equation}
Y(\omega) = \frac{1}{M} \sum_{m=0}^{M-1}{\exp\left[
                                 - \frac{jdb 2 \pi m}{T_{y}^{2}}(\omega + m\pi) 
                                \right]} 
    \label{eq_downsample} 
\end{equation} 

\subsection{Review of Filter-banks}  
Consider a two channel filter-bank shown in fig.\ref{fig_fb}. The input-output relation of this filter-bank can be written either in terms polyphase matrices, or in terms of modulation (or alias component) matrices \cite{bk_vaidyanathan_filterbank}, \cite{bk_vetterli_filterbank}.

Given a discrete signal $x(n)$, having ZT $X(z)$, its type-1 polyphase representation is obtained by writing $X(z)$ as $X(z) = X_{0}(z^{2}) + z^{-1}X_{1}(z^{2})$, where $X_{0}(z) = \sum_{n \in I}{x(2n)z^{-n}}$, $X_{1}(z) = \sum_{n \in I}{x(2n+1)z^{-n}}$. Similarly a type-2 polyphase representation is obtained by writing $X(z) = X_{0}(z^{2}) + zX_{1}(z^{2})$, where $X_{0}(z)$ is same as above and $X_{1}(z) = \sum_{n \in I}{x(2n-1)z^{-n}}$. Using these polyphase representations, it can be shown that 

\begin{equation}
\hat{{\bf X}}_{p}(z) = {\bf G}_{p}(z){\bf H}_{p}(z){\bf X}_{p}(z)
    \label{eq_poly_filterbank_z} 
\end{equation} 
where
\begin{eqnarray}
\hat{{\bf X}}_{p}(z) 
& \stackrel{def}{=} &
\left[
   \begin{array}{l}
     \hat{X}_{0}(z) \\
     \hat{X}_{1}(z)
   \end{array}
\right] \nonumber \\ 
{\bf X}_{p}(z) & \stackrel{def}{=}  &
\left[
   \begin{array}{l}
      X_{0}(z) \\
      X_{1}(z)
   \end{array}
\right] \nonumber \\
{\bf G}_{p}(z) & \stackrel{def}{=}  &
\left[
\begin{array}{ll}
G_{00}(z) & G_{10}(z) \\
G_{01}(z) & G_{11}(z)
\end{array}
\right] \nonumber \\
{\bf H}_{p}(z) & \stackrel{def}{=}  & 
\left[
\begin{array}{ll}
H_{00}(z) & H_{01}(z) \\
H_{10}(z) & H_{11}(z)
\end{array}
\right] 
    \label{eq_poly_filterbank_z_defn} 
\end{eqnarray} 
where the polyphase decompositions of analysis and synthesis filters are of different type (i.e., if $\{H_{0}(z),H_{1}(z)\}$ have type-1 polyphase decomposition, then $\{G_{0}(z),G_{1}(z)\}$ will have type-2 representation and vice versa), and the polyphase decomposition of $X(z)$ and $\hat{X}(z)$ are of same type.

The input-output relation can also be expressed in terms of modulation matrices as 

\begin{equation}
\hat{{\bf X}}_{m}(z) = \frac{1}{2} {\bf G}_{m}(z){\bf H}_{m}(z){\bf X}_{m}(z)
    \label{eq_mod_filterbank_z} 
\end{equation} 
where
\begin{eqnarray}
\hat{{\bf X}}_{m}(z) 
& \stackrel{def}{=} &
\left[
   \begin{array}{l}
     \hat{X}(z) \\
     \hat{X}(-z)
   \end{array}
\right] \nonumber \\ 
{\bf X}_{m}(z) & \stackrel{def}{=}  &
\left[
   \begin{array}{l}
      X(z) \\
      X(-z)
   \end{array}
\right] \nonumber \\
{\bf G}_{m}(z) & \stackrel{def}{=}  &
\left[
\begin{array}{ll}
G_{0}(z) & G_{1}(z) \\
G_{0}(-z) & G_{1}(-z)
\end{array}
\right] \nonumber \\
{\bf H}_{m}(z) & \stackrel{def}{=}  & 
\left[
\begin{array}{ll}
H_{0}(z) & H_{0}(-z) \\
H_{1}(z) & H_{1}(-z)
\end{array}
\right] 
    \label{eq_mod_filterbank_z_defn} 
\end{eqnarray} 

It can be seen that the polyphase matrices and the modulation matrices are related as 
\begin{equation}
\left[
\begin{array}{ll}
H_{0}(z) & H_{0}(-z) \\
H_{1}(z) & H_{1}(-z)
\end{array}
\right] 
= 
\left[
\begin{array}{ll}
H_{00}(z) & H_{01}(z) \\
H_{10}(z) & H_{11}(z)
\end{array}
\right] 
\left[
\begin{array}{ll}
1     &  1 \\
z^{-k} &  -z^{-k}
\end{array}
\right] 
    \label{eq_pm} 
\end{equation} 
where $k = 1$ for type-1 polyphase decomposition and $k=-1$ for type-2 polyphase representation.

A square polynomial matrix ${\bf H}(z)$ is said to be {\it paraunitary} (PU) if ${\bf H}(z)\tilde{{\bf H}}(z) = dI$ for some $d > 0$. It can be seen from (\ref{eq_pm}) that if polyphase matrix of analysis filters is PU, then its alias component matrix will also be PU.

If ${\bf H}_{m}(z)$ is PU, then by making ${\bf G}_{m}(z) = z^{-K}\tilde{{\bf H}}_{m}(z)$, where $K$ is an integer appropriate enough to make the synthesis filters causal, $\hat{X}(z) = z^{-K}X(z)$ can be achieved, i.e. the output is delayed version of the input, and perfection reconstruction can be achieved. 

If ${\bf H}_{m}(z)$ is PU, then it can be seen that $H_{0}(z)$ satisfies what is called {\it power-symmetry} property, i.e.
 \begin{equation}
|H(e^{j\omega})|^{2} + |H(e^{j(\omega + \pi)})|^{2} = d 
    \label{eq_ps} 
\end{equation} 
Further $H_{1}(z)$ is related to $H_{0}(z)$ as 
\begin{equation}
H_{1}(z) = c z^{-L}\tilde{H}_{0}(z), \hspace{3mm} |c| = 1, \hspace{3mm} L=odd
    \label{eq_h1_h2_z} 
\end{equation} 

If ${\bf H}_{m}(z)$ is PU, then making ${\bf G}_{m}(z) = z^{-K}{\bf H}_{m}(z)$ will give us PRFB. This gives

\begin{eqnarray}
G_{0}(z) = z^{-K}\tilde{H}_{0}(z) \nonumber \\
G_{1}(z) = z^{-K}\tilde{H}_{1}(z) 
    \label{eq_gh_z} 
\end{eqnarray} 

Thus all the filters of a two channel PU filter bank can be obtained by designing a power symmetric filter $H_{0}(z)$. 

If $H_{0}(z)$ is power-symmetric, then $H_{0}(z)\tilde{H}_{0}(z)$ is an half-band filter. Thus a power-symmetric filter can be designed by designing a half-band filter and computing an appropriate spectral factor (Please refer \cite{bk_vaidyanathan_filterbank} for further details.).

\section{Convolution and Delay in the DTLCT domain}
\label{sec_basic_operations}
   Filtering operation in sub-band decomposition is defined as taking product of the DFT(or ZT) of the signal with the filter. In time domain this corresponds to convolution between the signal and the filter. In continuous FrFT and LCT domain also various convolution theorems have been derived \cite{art_bing_convolution_lct},\cite{art_almeida_conv_frft}.

For LCT sub-band decomposition, it is convenient to define  convolution of $x(n)$ and $h(n)$ as

\begin{equation}
h(n)*x(n) \stackrel{def}{=} \sum_{k=-\infty}^{\infty}{h(k)x(n-k)e^{-\frac{jaT^{2}}{b}k(n-k)}} 
    \label{eq_lct_conv} 
\end{equation} 

Thus if $y(n) = h(n)*x(n)$, then 
\begin{equation}
Y(\omega) = H(\omega)X(\omega)e^{-\frac{j db \omega^{2}}{2 T^{2}}}
    \label{eq_lct_conv2} 
\end{equation} 

With this definition of convolution, it is convenient to define a new {\it delay} operator $D[.]$ on a signal $x(n)$ as 

\begin{equation}
D[x](n) = x(n-1)e^{\frac{j a T^{2}}{2b} (-2n + 1)}.
    \label{eq_delay_defn} 
\end{equation} 

The properties of $D[.]$ are summarized in the following lemma,

\begin{lemma}
The operator $D[.]$ has the following properties,
\begin{enumerate}
    \item $D^{l}[D^{k}[x]](n) = D^{l+k}[x](n)$ for all $k,l \in I$.
    \item If $y(n) = D^{k}[x](n)$ then $Y(\omega) = e^{-jk\omega}X(\omega)$.
    \item If $y(n) = x(n)*h(n)$ then $D^{l}[x](n)*D^{k}[h](n) = D^{l+k}[y](n)$.
\end{enumerate} 

\end{lemma} 

\begin{proof}
The first two results are straight-forward. To prove the last result let
\begin{eqnarray}
D^{l}[x](n)*D^{k}[h](n)     
     &=& \sum_{r \in I}
         {D^{l}[x](r)D^{k}[h](n-r)e^{-\frac{jaT^{2}}{b}r(n-r)}} \nonumber \\
     &=& \sum_{r \in I}
         {x(r-l)e^{\frac{jaT^{2}(-2rl+l^{2})}{2b} }
          h(n-r-k) e^{\frac{jaT^{2}(-2(n-r)k+k^{2})}{2b} } 
          e^{-\frac{jaT^{2}}{b}r(n-r)}}
    \label{eq_delay} 
\end{eqnarray} 
By re-arranging the index and the power of $e$, we can obtain $D^{l}[x](n)*D^{k}[h](n) = D^{k+l}[y](n)$.
\end{proof} 

\section{A Two Channel LCT Filter Bank}
\label{sec_filter_bank}

Consider the filter bank in fig.\ref{fig_fb}, with convolution as defined in (\ref{eq_lct_conv}). We are interested in writing the input-output relation of this filter bank in terms of polyphase and modulation matrices, if ZT (or DFT) is replaced by DTLCT.  

%\begin{figure}
%\centerline{\epsfig{figure=filterbank.eps,width=0.66\textwidth}}
%\caption{A two channel filter bank.}
%\label{fig_fb} 
%\end{figure}

\begin{figure}[t]
  \centering
  \includegraphics[scale=0.4]{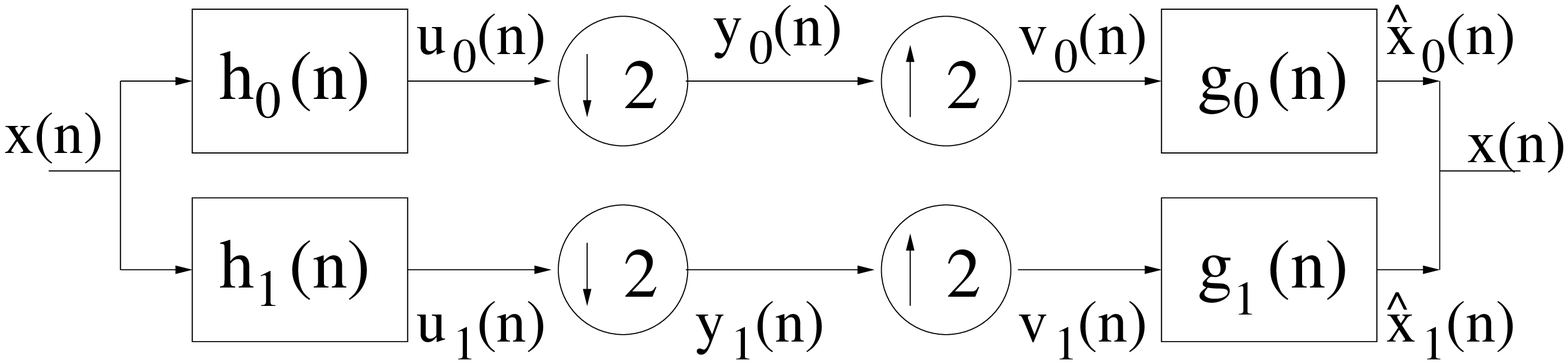}
  \caption{The two channel filter bank.}
  \label{fig_fb} 
\end{figure}

\subsection{Polyphase Domain Analysis}

To define a type-1 polyphase decomposition of $x(n)$, let $x(n) = x_{e}(n) + x_{o}(n)$, where 
\begin{equation}
x_{e}(n) = \left\{
\begin{array}{ll}
x(n) & n = even \\
0    & n = odd
\end{array}
\right.
    \label{eq_xen_defn} 
\end{equation}  
and $x_{o}(n)$ is similarly defined. Let $\bar{x}_{o}(n) = D^{-1}[x_{o}](n)$, then $X(\omega) = X_{e}(\omega) + e^{j\omega}\bar{X}_{o}(\omega)$. Let $x_{0}(n) = x_{e}(n) \downarrow 2$ and $x_{1}(n) = \bar{x}_{e}(n) \downarrow 2$, then type-1 polyphase decomposition of $x(n)$ is given as $X(\omega) = X_{0}(2\omega) + e^{j\omega}X_{1}(2\omega)$.

Similarly, we can obtain type-2 polyphase decomposition of $x(n)$ by putting $\bar{x}_{o}(n) = D[x_{o}](n)$. In this case $X(\omega) = X_{0}(2\omega) + e^{-j\omega}X_{1}(2\omega)$. Note that the sampling time of the polyphase components $T_{x_{0}} = T_{x_{1}} = 2T_{x}$.

\begin{figure}[tpb]
  \centering
  \includegraphics[scale=0.4]{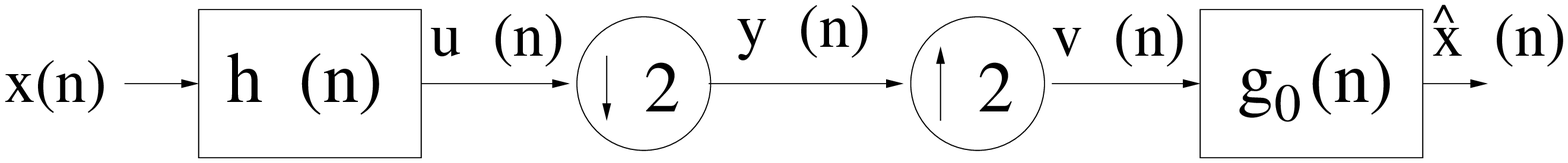}
  \caption{A filter bank channel.}
  \label{fig_ch} 
\end{figure}

The input-output relation of an analysis filter bank channel, shown in fig.\ref{fig_ach}, is given in terms of polyphase components as 

\begin{figure}[t]
  \centering
  \subfigure[]{
    \includegraphics[scale=0.3]{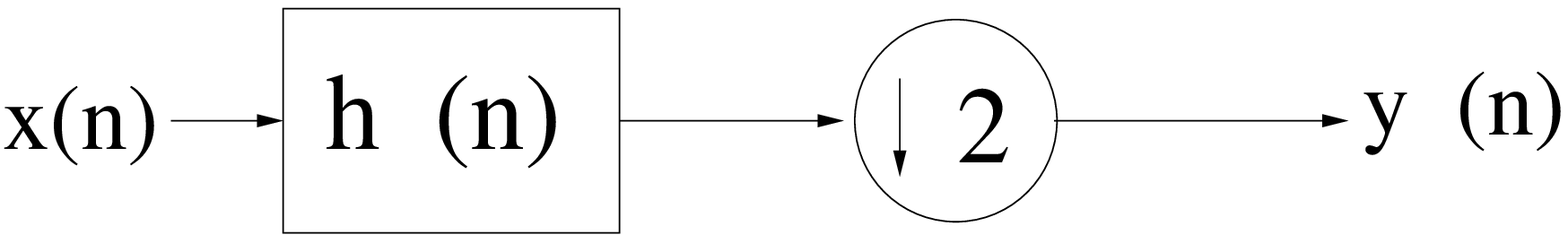}
    \label{fig_ach} 
  }
  \hspace{0.5in}
  \subfigure[]{
    \includegraphics[scale=0.3]{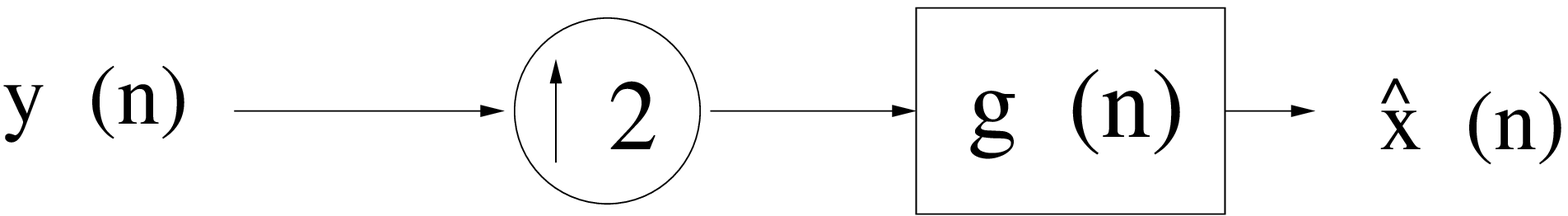}
    \label{fig_sch} 
  }
  \caption{Filter bank channels: (a) Analysis filter bank channel (b) Synthesis filter bank channel}
\end{figure}

%\begin{figure}
%\centerline{\epsfig{figure=ach.eps,width=0.66\textwidth}}
%\caption{A channel of the analysis filter bank.}
%\label{fig_ach} 
%\end{figure}

%\begin{figure}
%\centerline{\epsfig{figure=sch.eps,width=0.66\textwidth}}
%\caption{A channel of the synthesis filter bank.}
%\label{fig_sch} 
%\end{figure}

\begin{lemma}
If $y(n) = (h(n)*x(n)) \downarrow 2$, then

\begin{equation}
Y(\omega) = e^{\frac{-jdb \omega^{2}}{T_{y}^{2}}}
            \left[
               \begin{array}{ll}
                 H_{0}(\omega)  &  H_{1}(\omega)
               \end{array}
            \right] 
            \left[
               \begin{array}{l}
                 X_{0}(\omega) \\
                 X_{1}(\omega)
               \end{array}
            \right] 
    \label{eq_poly_analysis} 
\end{equation}
where the polyphase components of $h(n)$ and $x(n)$ are of different type.
\end{lemma} 
\begin{proof}
Let $u(n) = h(n)*x(n)$, then,
\begin{eqnarray}
u_{e}(n) &=& x_{e}(n)*h_{e}(n) + x_{o}(n)*h_{o}(n) \nonumber \\
         &=& x_{e}(n)*h_{e}(n) + D[x_{o}](n)*D^{-1}[h_{o}](n) \nonumber \\
         &=& x_{e}(n)*h_{e}(n) + \bar{x}_{o}(n)*\bar{h}_{o}](n)
    \label{eq_poly_analysis_proof} 
\end{eqnarray} 

Taking DTLCT of both the sides,
\begin{eqnarray}
U_{e}(\omega)
 &=& e^{\frac{-jdb \omega^{2}}{T_{u}^{2}}}(H_{e}(\omega)X_{e}(\omega) + 
     \bar{H}_{o}(\omega)\bar{X}_{o}(\omega)) \nonumber \\
 &=& e^{\frac{-jdb \omega^{2}}{T_{u}^{2}}}(H_{0}(2\omega)X_{0}(2\omega) + 
     H_{1}(2\omega)X_{1}(2\omega)) 
    \label{eq_poly_analysis_proof_2} 
\end{eqnarray} 

By noting that $Y(\omega) = U_{e}(\omega/2)$, and $T_{y} = 2T_{u}$, and substituting it in (\ref{eq_poly_analysis_proof_2}), we obtain the required result.
\end{proof} 

Similarly the input-output relation of a synthesis filter bank channel, shown in fig.\ref{fig_sch}, is given in terms of polyphase components as-

\begin{lemma}
If $\hat{x}(n) = g(n)*(y(n) \uparrow 2)$, then

\begin{equation}
\left[
   \begin{array}{l}
     \hat{X}_{0}(\omega) \\
     \hat{X}_{1}(\omega)
   \end{array}
\right] 
= e^{\frac{-jdb\omega^{2}}{2T_{y}^{2}}}
\left[
\begin{array}{l}
G_{0}(\omega) \\
G_{1}(\omega)
\end{array}
\right] Y(\omega)
    \label{eq_poly_synthesis} 
\end{equation} 
where the polyphase components of both $g(n)$ and $x(n)$ are of same type.
\end{lemma} 
\begin{proof}
Let $v(n) = y(n) \uparrow 2$. Then $\hat{x}_{e}(n) = g_{e}(n)*v(n)$ and $\hat{x}_{o}(n) = g_{o}(n)*v(n)$. This gives $D[\hat{x}_{o}](n) = D[g_{o}]{n}*v(n)$. Taking LCT of both the sides, we have
\begin{equation}
\left[
   \begin{array}{l}
      \hat{X}_{e}(\omega) \\
      \bar{\hat{X}}_{o}(\omega)
   \end{array}
\right] 
= e^{\frac{-jdb \omega^{2}}{2T_{v}^{2}}}
\left[
   \begin{array}{l}
      G_{e}(\omega) \\
      \bar{G}_{o}(\omega)
   \end{array}
\right] V(\omega)
    \label{eq_synthesis_proof_2} 
\end{equation} 
By noting that $T_{y} = 2T_{v}$, $V(\omega) = Y(2\omega)$, $X_{e}(\omega) = X_{0}(2\omega)$ and so on, and substituting them in (\ref{eq_synthesis_proof_2}) we get the desired result. 
\end{proof} 

Using the above two lemma, The filter bank equation for the filter-bank in fig.\ref{fig_fb}  is

\begin{equation}
{\bf \hat{X}}_{p}(\omega) =  e^{\frac{-jdb\omega^{2}}{T_{x_{0}}^{2}}}
                         {\bf G}_{p}(\omega){\bf H}_{p}(\omega)
                         {\bf X}_{p}(\omega) 
    \label{eq_poly_filterbank} 
\end{equation} 
where
\begin{eqnarray}
{\bf X}_{p}(\omega) & = &
\left[
   \begin{array}{l}
      X_{0}(\omega) \\
      X_{1}(\omega)
   \end{array}
\right] \nonumber \\
{\bf \hat{X}}_{p}(\omega) & = &
\left[
   \begin{array}{l}
     \hat{X}_{0}(\omega) \\
     \hat{X}_{1}(\omega)
   \end{array}
\right] \nonumber \\
{\bf G}(\omega) & = &
\left[
\begin{array}{ll}
G_{00}(\omega) & G_{10}(\omega) \\
G_{01}(\omega) & G_{11}(\omega)
\end{array}
\right] \nonumber \\
{\bf H}(\omega) & = &
\left[
\begin{array}{ll}
H_{00}(\omega) & H_{01}(\omega) \\
H_{10}(\omega) & H_{11}(\omega)
\end{array}
\right]
    \label{eq_poly_filterbank_def} 
\end{eqnarray} 

\subsection{Modulation Domain Analysis}
To write the filter bank equation in terms of alias component matrices, it is preferable to consider a channel of the filter bank as shown in fig.\ref{fig_ch} . The input-output relation for this channel is given as follows-

\begin{lemma}
  For the channel in fig.\ref{fig_ch} , the input and output are related by
\begin{equation}
\hat{X}(\omega) = \frac{1}{2} \sum_{m=0}^{1}{e^{\frac{-jdb(\omega+m\pi)^{2}}{T_{x}^{2}}}
                              G(\omega)H(\omega+m\pi)X(\omega+m\pi)} 
    \label{eq_modulation_lemma_1}
\end{equation} 

Moreover, $\hat{X}(\omega + \pi)$ is given by

\begin{equation}
\hat{X}(\omega + \pi) = \frac{1}{2} \sum_{m=0}^{1}{e^{\frac{-jdb(\omega+m\pi)^{2}}{T_{x}^{2}}}
                              G(\omega+\pi)H(\omega+m\pi)X(\omega+m\pi)} 
    \label{eq_modulation_lemma_2}
\end{equation} 

\end{lemma}
\begin{proof}
By using the following relations
\begin{eqnarray}
U(\omega) &=& e^{-\frac{jdb\omega^{2}}{2T_{x}^{2}}}F(\omega)X(\omega) \nonumber \\
Y(\omega) &=& \frac{1}{2} \sum_{m=0}^{1}{exp(-\frac{jdb 2m\pi(\omega+m\pi)}{T_{y}^{2}})
                            U(\frac{\omega}{2}+m\pi)} 
                            \hspace{3mm} T_{y} = 2T_{x} \nonumber \\
V(\omega) &=& Y(2\omega) \hspace{3mm}    T_{v} = T_{y}/2 = T_{x} \nonumber \\
\hat{X}(\omega) &=& e^{-\frac{jdb\omega^{2}}{2T_{x}^{2}} } G(\omega)V(\omega)
    \label{} 
\end{eqnarray} 
and some algebraic manipulation we can obtain (\ref{eq_modulation_lemma_1}).

To prove (\ref{eq_modulation_lemma_2}), we note that using (\ref{eq_modulation_lemma_1}), we have

\begin{equation}
\hat{X}(\omega+\pi) = \frac{1}{2} \sum_{0}^{1}{e^{\frac{-jdb(\omega+(m+1)\pi)^{2}}{T_{x}^{2}}}
                              G(\omega+\pi)H(\omega+(m+1)\pi)X(\omega+(m+1)\pi)} 
    \label{eq_modulation_lemma_12}
\end{equation} 
Further noting that for any $X(\omega)$
\begin{equation}
X(\omega+2\pi) = e^{\frac{jdb 2\pi(\omega+\pi)}{T_{x}^{2}}}X(\omega)
    \label{eq_2pi_relation} 
\end{equation} 
and substituting from (\ref{eq_2pi_relation}) into (\ref{eq_modulation_lemma_12}) we get (\ref{eq_modulation_lemma_2}).
\end{proof}  

Using the above lemma we can write the following input-output relation for the filter bank shown in fig.\ref{fig_fb} .
\begin{equation}
{\bf \hat{X}}_{m}(\omega) = \frac{1}{2} {\bf G}_{m}(\omega){\bf H}_{m}(\omega)
                            {\bf A}(\omega) {\bf X}_{m}(\omega) 
    \label{eq_mod_filterbank} 
\end{equation} 

\begin{eqnarray}
{\bf X}_{m}(\omega) &=&
\left[
   \begin{array}{l}
      X(\omega) \\
      X(\omega + \pi) 
   \end{array}
\right] \nonumber \\
{\bf \hat{X}}_{m}(\omega) &=&
\left[
   \begin{array}{l}
      \hat{X}(\omega) \\
      \hat{X}(\omega + \pi) 
   \end{array}
\right] \nonumber \\
{\bf G}_{m}(\omega) &=&
\left[
   \begin{array}{ll}
      G_{0}(\omega)       & G_{1}(\omega) \\
      G_{0}(\omega + \pi) & G_{1}(\omega + \pi)
   \end{array}
\right] \nonumber \\
{\bf H}_{m}(\omega) &=&
\left[
   \begin{array}{ll}
      H_{0}(\omega)       & H_{0}(\omega + \pi) \\
      H_{1}(\omega)       & H_{1}(\omega + \pi)
   \end{array}
\right] \nonumber \\
{\bf A}(\omega) &=& 
\left[
   \begin{array}{ll}
      e^{-\frac{jdb\omega^{2}}{T_{x}^{2}}} & 0 \\
      0       & e^{-\frac{jdb(\omega+\pi)^{2}}{T_{x}^{2}}}
   \end{array}
\right] 
    \label{eq_mod_filterbank_def} 
\end{eqnarray} 

Using the relation $X(\omega + 2\pi) = e^{j \frac{2\pi d b}{T^{2}}(\omega + \pi)}X(\omega)$, it can be shown that
\begin{equation}
{\bf X}_{m}(\omega) = {\bf B}(\omega){\bf C}(e^{j\omega}){\bf X}_{p}(\omega)    
    \label{eq_mod_poly} 
\end{equation} 
where
\begin{eqnarray}
{\bf B}(\omega) &=&
\left[
\begin{array}{ll}
   1 & 0 \\
   0 & e^{j \frac{2\pi d b}{T^{2}}(2\omega + \pi)}
\end{array}
\right] \nonumber \\
{\bf C}(\omega) &=&
\left[
\begin{array}{ll}
   1 & e^{j\omega} \\
   1 & -e^{j\omega}
\end{array}
\right] 
    \label{eq_mod_poly_def} 
\end{eqnarray} 

\section{Two Channel Paraunitary Filter Banks}
\label{sec_paraunitary}
A square matrix ${\bf H}(\omega)$ is defined to be a {\it paraunitary} matrix if ${\bf H}_{*}^{T}(\omega){\bf H}(\omega) = dI$ for some $d > 0$ and for all $\omega$. Let ${\bf H}_{m}(\omega)$ in (\ref{eq_mod_filterbank}) be PU, then by setting 
\begin{equation}
{\bf G}_{m}(\omega) = e^{-jK\omega}{\bf A}_{*}(\omega){\bf H}_{m*}^{T}(\omega)   
    \label{eq_gh} 
\end{equation} 
where $K$ is an integer (needed to make the synthesis filters causal), we can obtain perfect reconstruction. 

We now show that as in the case of ZT, ${\bf H}(\omega)$ can be made PU if $H_{0}(\omega)$ is power-symmetric, and $H_{1}(\omega)$ can be obtained from $H_{0}(\omega)$ by inspection.
 
\begin{lemma}
If $H_{0}(\omega)$ is power symmetric, i.e.
\begin{equation}
\left|H_{0}(\omega)\right|^{2} + \left|H_{0}(\omega+\pi)\right|^{2} = 1,
    \label{eq_omega_power_symmetry} 
\end{equation} 
and $H_{1}(\omega)$ is such that
\begin{equation}
H_{1}(\omega) = e^{j\left(\frac{db \omega (\omega + \pi)}{T^{2}} + (2L+1)\omega \right)}H_{0*}(\omega+\pi),
    \label{eq_h0h1} 
\end{equation} 
for an integer $L$, then ${\bf H}(\omega)$ is PU.
\label{th_h1_form_h0}
\end{lemma} 
\begin{proof}
By expanding each element in ${\bf H}_{m}(\omega){\bf H}_{m*}^{T}(\omega)$,substituting from (\ref{eq_omega_power_symmetry}) and (\ref{eq_h0h1}), and using the relation $H_{i}(\omega + 2\pi) = e^{j \frac{2\pi d b}{T^{2}}(\omega + \pi)}H_{i}(\omega)$, where $i \in \{0,1\}$, we can get ${\bf H}_{m}(\omega){\bf H}_{m*}^{T}(\omega) = I$.
\end{proof} 

The above lemma suggests that just like in the ZT filter bank, problem of designing a two channel PU LCT filter bank can be reduced to the problem of designing a power-symmetric filter, and rest of the filters could be obtained from this filter. However a power-symmetric filter in LCT domain can be obtained from a power-symmetric filter in ZT domain, as shown in the following lemma,

\begin{lemma}
\label{th_power_symmetric}
Let $H_{0}(\omega) = DTLCT(h_{0})$. Let $\ddot{h}_{0}(n)$ be defined as $\ddot{h}_{0}(n) = h_{0}(n)e^{\frac{jan^{2}T^{2}}{2b}}$. Let $\ddot{H}_{0}(e^{j\omega}) = FT(\ddot{h}_{0})$. If $\ddot{H}_{0}(e^{j\omega})$ is power-symmetric, then $H_{0}(\omega)$ will also be power-symmetric.
\end{lemma} 
\begin{proof}
$H_{0}(\omega)$ can be written as
\begin{eqnarray}
H_{0}(\omega) 
&=& \sum_{n}{h_{0}(n) e^{\frac{jan^{2}T^{2}}{2b}} e^{-jn\omega} 
             e^{\frac{jdb\omega^{2}}{2T^{2}} }}  \nonumber \\
&=& e^{\frac{jdb\omega^{2}}{2T^{2}} }\ddot{H}_{0}(e^{j\omega})
    \label{eq_h0_h0dot} 
\end{eqnarray} 

From this it can be seen that 

\begin{equation}
H_{0}(\omega)H_{0}^{*}(\omega) = \ddot{H}_{0}(e^{j\omega})\ddot{H}_{0}^{*}(e^{-j\omega})
    \label{eq_pc_cond} 
\end{equation} 
and so if $\ddot{H}_{0}(e^{j\omega})$ is power-symmetric, then $H_{0}(\omega)$ will also be power-symmetric. 
\end{proof} 

\section{An Example}
In this section we take an example which serves as a demonstration of the LCT filter bank. In this example we take the LCT to be FrFT at an angle $\pi/4$. The $[a,b,c,d]$ matrix will thus be
\begin{equation}
\left[
\begin{array}{ll}
a & b \\
c & d
\end{array}
\right] = \left[
\begin{array}{ll}
\cos(\pi/4) & \sin(\pi/4) \\
-\sin(\pi/4) & \cos(\pi/4) 
\end{array}
\right]  
    \label{eq_abcd} 
\end{equation} 

The input signal $x(n)$ is such that its LCT $X(\omega)$ has peaks at$ \{\frac{30 \pi}{512}, \frac{100 \pi}{512}, \frac{412 \pi}{512},\frac{482 \pi}{512}\}$ , as shown in fig.\ref{fig_input}. The sampling time of the input signal $x(n)$ is taken as $T = 0.05$.

%\begin{figure}
%\centerline{\epsfig{figure=lctx.eps,width=0.66\textwidth}}
%\caption{Magnitude plot of the input signal $X(\omega)$ to the filter bank.}
%\label{fig_input} 
%\end{figure}

\begin{figure}[t]
  \centering
  \includegraphics[scale=0.75]{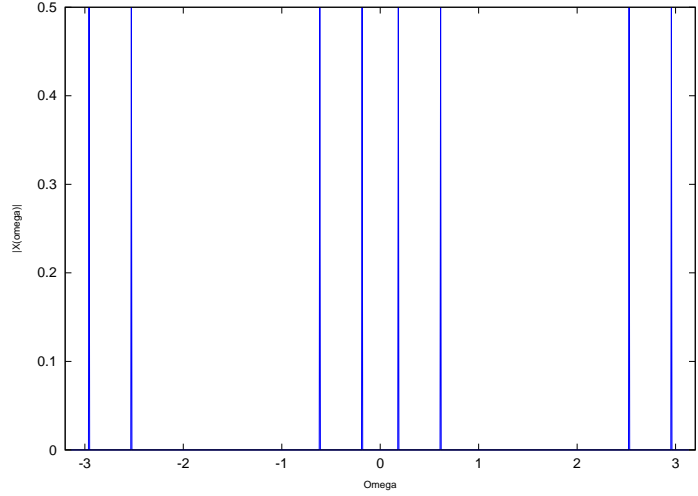}
  \caption{Magnitude of the input signal $X(\omega)$.}
  \label{fig_input} 
\end{figure}

An LCT power-symmetric filter is generated by applying lemma \ref{th_power_symmetric} on a ZT power-symmetric filter. In this example, the ZT power symmetric filter $h(n)$ is taken from \cite{bk_vaidyanathan_filterbank},example 5.3.2. %The coefficients of $h(n)$ are given in table-\ref{tbl_h}, and 
The magnitude response is shown in fig(\ref{fig_fth}). Using lemma \ref{th_power_symmetric}, the LCT power symmetric filter $h_{0}(n)$ is given by
\begin{equation}
h_{0}(n) = h(n)e^{-\frac{j}{2}\frac{an^{2}T^{2}}{b}}
    \label{eq_h0_from_h} 
\end{equation} 

\begin{figure}[t]
  \centering
  \includegraphics[scale=0.75]{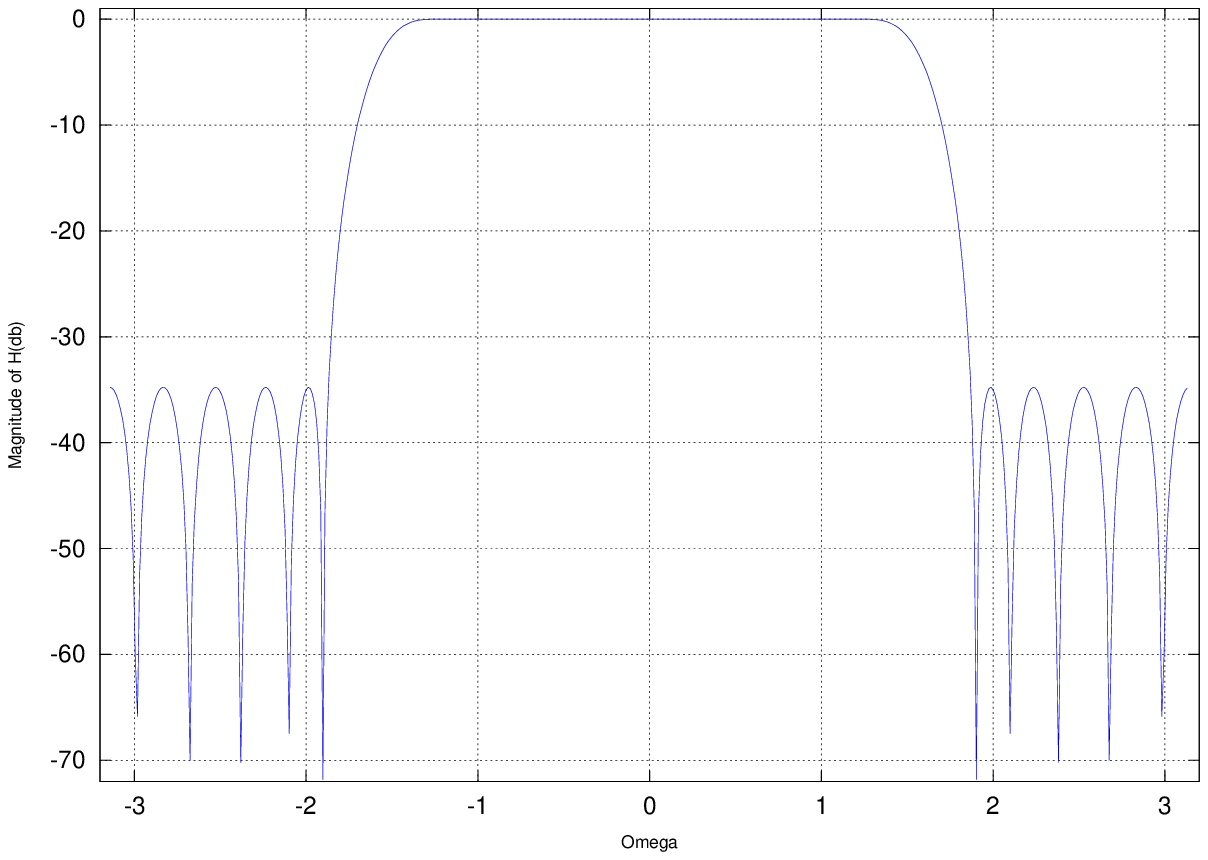}
  \caption{Magnitude(dB) of the power symmetric filter $H(e^{j\omega})$, that is used to generate $H_{0}(\omega)$}
  \label{fig_fth} 
\end{figure}

%\begin{figure}
%\centerline{\epsfig{figure=fth.eps,width=0.66\textwidth}}
%\caption{Magnitude plot(in db) of the power symmetric filter $H(e^{j\omega})$, that is used to generate $H_{0}(\omega)$}
%\label{fig_fth} 
%\end{figure}

%\begin{table}
%\begin{tabular}{|c|c|}
%\hline
%n & H(n) \\
%\hline
%00 & 0.1605476    \\
%01 & 0.4156381    \\
%02 & 0.4591917    \\
%03 & 0.1487153    \\
%04 & -0.1642893   \\
%05 & -0.1245206   \\
%06 & 0.08252419   \\
%07 & 0.08875733   \\
%08 & -0.05080163  \\
%09 & -0.06084593  \\
%10 & 0.03518087   \\
%11 & 0.03989182   \\
%12 & -0.02561513  \\
%13 & -0.02440664  \\
%14 & 0.01860065   \\
%15 & 0.01354778   \\
%16 & -0.01308061  \\
%17 & -0.007449561 \\
%18 & 0.01293440   \\
%19 & -0.004995356 \\
%\hline
%\end{tabular}
%\caption{Coefficients of the power-symmetric filter $H(n)$ used to generate LCT filter bank.} 
%\label{tbl_h}
%\end{table}

Filter $h_{1}(n)$ is obtained from $h_{0}(n)$ by applying lemma \ref{th_h1_form_h0}. This gives $h_{1}(n)$ as
\begin{equation}
h_{1}(k) = h_{0*}(N-k)e^{j\phi(k)}, k \in \{0..N\}
    \label{eq_h1_from_h0_2} 
\end{equation}   
where
\begin{equation}
\phi(k) = -\frac{aT^{2}}{2b}\left[(N-k)^{2}+k^{2}\right] + (N-k)\pi -\frac{db}{2T^{2}}\pi^{2}  
    \label{eq_phi_k} 
\end{equation} 

The magnitude response of the LCT of analysis filters is shown in fig.\ref{fig_lcta}.

\begin{figure}[t]
  \centering
  \subfigure[]{
    \includegraphics[scale=0.6]{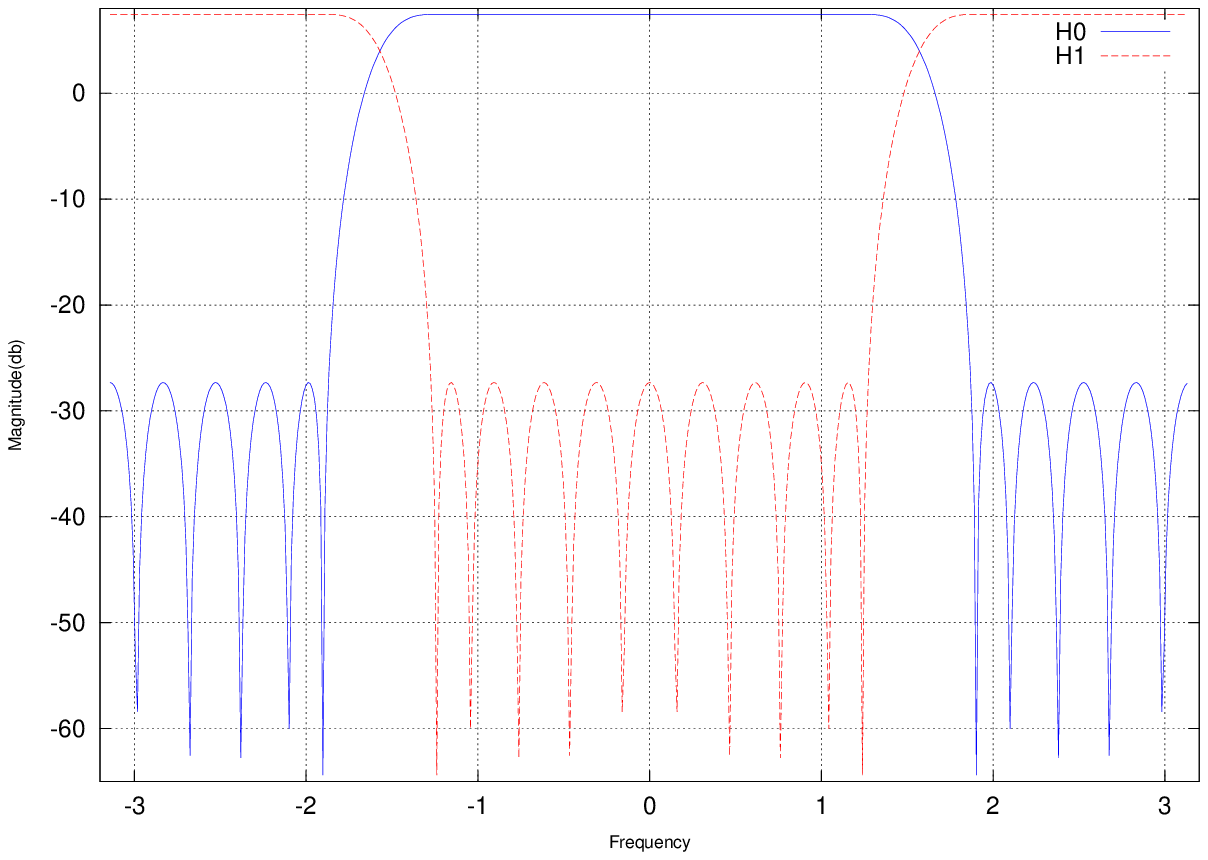}
    \label{fig_lcta} 
  }
  \hspace{0.5in}
  \subfigure[]{
    \includegraphics[scale=0.6]{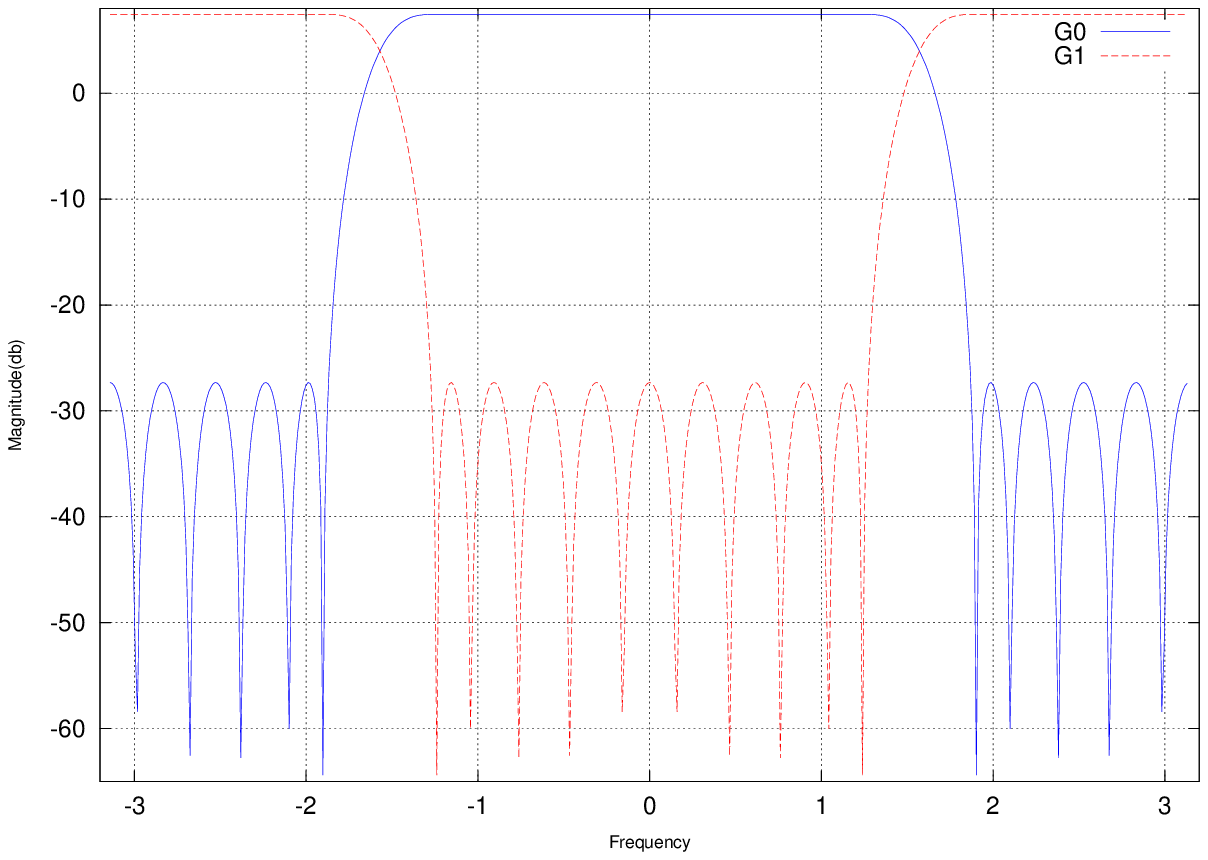}
    \label{fig_lcts} 
  }
  \caption{Magnitude(dB) of the LCT of various filters used in the LCT filter bank: (a) Analysis filters $\{H_{0}(\omega),H_{1}(\omega)\}$ (b) Synthesis filters $\{G_{0}(\omega),G_{1}(\omega)\}$}
\end{figure}

%\begin{figure}
%\centerline{\epsfig{figure=lcta.eps,width=0.66\textwidth}}
%\caption{Magnitude responce(in db), of the LCT of analysis filters $H_{0}(\omega),H_{1}(\omega)$, generated from $H(e^{j\omega})$.}
%\label{fig_lcta} 
%\end{figure}

Finally the synthesis filters $\{G_{0}(\omega),G_{1}(\omega)\}$ are obtained from $\{H_{0}(\omega),H_{1}(\omega)\}$ using (\ref{eq_gh}), which gives
\begin{equation}
g_{i}(k) = h_{i*}(N-k)e^{-\frac{j}{2}\frac{aT^{2}}{b}\left[(N-k)^{2}+k^{2}\right]}, k \in \{0..N\}, i \in {0,1}
    \label{eq_gh2} 
\end{equation} 

The magnitude response of the LCT of synthesis filters is shown in fig.\ref{fig_lcts}.

%\begin{figure}
%\centerline{\epsfig{figure=lcts.eps,width=0.66\textwidth}}
%\caption{Magnitude responce(in db), of the LCT of synthesis filters $G_{0}(\omega),G_{1}(\omega)$.}
%\label{fig_lcts} 
%\end{figure}

The LCT of analysis filter output $\{y_{0}(n),y_{1}(n)\}$ are shown in fig.\ref{fig_lcty0} and fig.\ref{fig_lcty1}. The peaks at $\{\frac{30 \pi}{256}, \frac{100 \pi}{256}\}$ in $Y_{0}(\omega)$ correspond to peaks at $\{\frac{30 \pi}{512}, \frac{100 \pi}{512}\}$ in $X(\omega)$. Similarly peaks at $\{\frac{30 \pi}{256}, \frac{100 \pi}{256}\}$ in $Y_{1}(\omega)$ correspond to peaks at $\{\frac{482\pi}{512}, \frac{412 \pi}{512}\}$ in $X(\omega)$. 

\begin{figure}[t]
  \centering
  \subfigure[]{
    \includegraphics[scale=0.6]{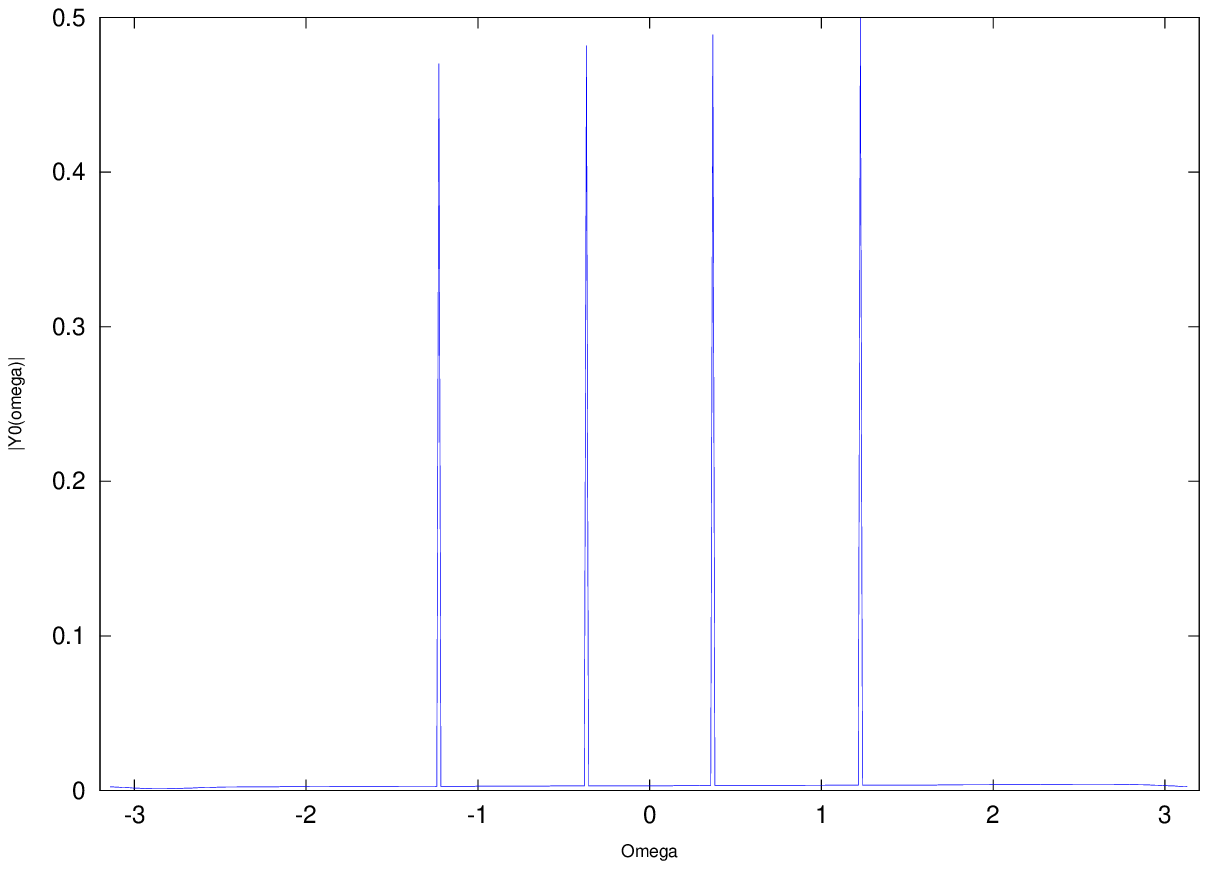}
    \label{fig_lcty0} 
  }
  \hspace{0.5in}
  \subfigure[]{
    \includegraphics[scale=0.6]{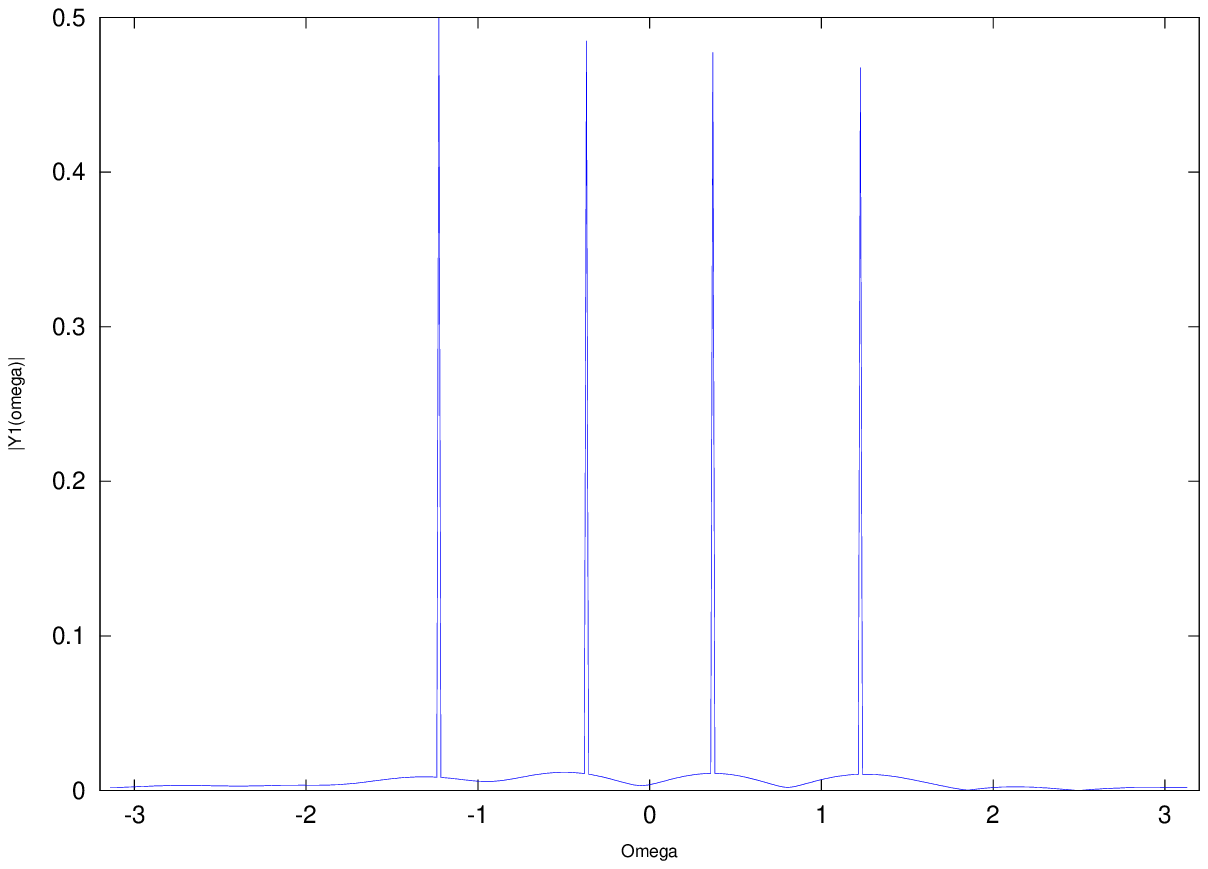}
    \label{fig_lcty1} 
  }
  \caption{Magnitude of analysis filter bank output: (a) $Y_{0}(\omega)$ (b) $Y_{1}(\omega)$.}
\end{figure}

%\begin{figure}
%\centerline{\epsfig{figure=lcty0.eps,width=0.66\textwidth}}
%\caption{Magnitude plot of the analysis filter bank output $Y_{0}(\omega)$.}
%\label{fig_lcty0} 
%\end{figure}

%\begin{figure}
%\centerline{\epsfig{figure=lcty1.eps,width=0.66\textwidth}}
%\caption{Magnitude plot, analysis filter bank output $Y_{1}(\omega)$.}
%\label{fig_lcty1} 
%\end{figure}

The signal is then reconstructed by passing $\{y_{0}(n),y_{1}(y)\}$ from synthesis filter. The LCT of the synthesis filter output $\hat{X}(\omega)$ is shown in fig.\ref{fig_lcthatx}. It is observed that the output matches with input delayed by $N$ samples.  

\begin{figure}[t]
  \centering
  \includegraphics[scale=0.75]{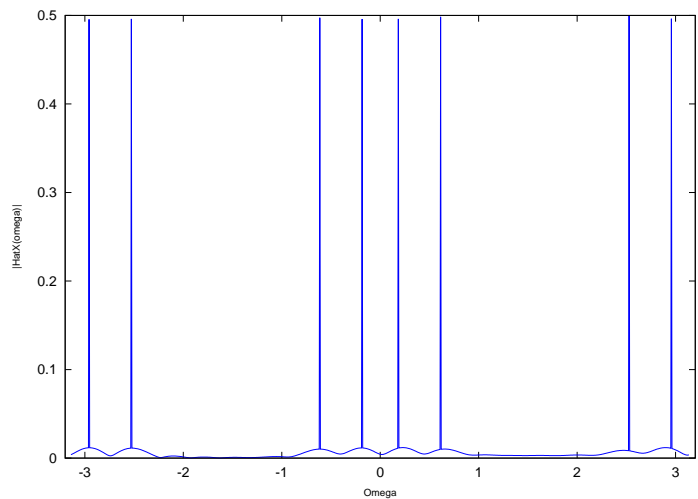}
  \caption{Magnitude of the filter bank output $\hat{X}(\omega)$.}
  \label{fig_lcthatx} 
\end{figure}

%\begin{figure}
%\centerline{\epsfig{figure=lcthatx.eps,width=0.66\textwidth}}
%\caption{Magnitude plot of the filter bank output $\hat{X}(\omega)$.}
%\label{fig_lcthatx} 
%\end{figure}

\section{Conclusion}
\label{sec_conclude}

In this paper a two channel paraunitary filter bank based on linear canonical transform has been developed. This kind of filter bank can be useful in sub-band decomposition of the signals that are not band limited in the Fourier domain, but band limited in an LCT domain. Input-output relations for such a filter bank in polyphase and modulation domain are derived. It is also shown that such a filter bank can be designed by using standard design procedure for power-symmetric filters in Fourier domain.

The future work in this direction may include generalization of LCT based filter bank to more than two channels, and development of LCT based cosine modulated filter banks. 

\pagebreak

%\appendix 
%\appandix 

%Appendix.

\bibliographystyle{ieeetr}
%\bibliographystyle{plain}
%\bibliography{./Bibl.bib}
%\include{lct_filter_bank.bbl}

%%%%%\bibliography{bib-file}  % commented if *.bbl file included, as seen below

\end{document}